\documentclass{article}
\usepackage[utf8]{inputenc}

\usepackage[nottoc]{tocbibind}
 \usepackage{verbatim} 
 \usepackage{amsmath}
\usepackage{amssymb}
\usepackage{amsmath}
\usepackage{amsthm}
\usepackage{bbm}
\usepackage{bm}
\usepackage{authblk}
\usepackage{graphicx}
\usepackage{subcaption}
\usepackage{caption}
\usepackage{color}
\usepackage[T1]{fontenc}
\usepackage[toc,page]{appendix}
\newcommand{\R}{\mathbb{R}}
\newcommand{\p}{\partial}
\usepackage[margin=1.3in]{geometry}

\title{Analytical formulation for multidimensional continuous opinion models}

\date{}
\author[1,2]{Luc\'ia Pedraza \thanks{lpedraza@df.uba.ar}}
\author[3]{Juan Pablo Pinasco}
\author[3]{Nicolas Saintier}
\author[1,2]{Pablo Balenzuela}

\affil[1]{\small{Departamento de F\'isica, Facultad de Ciencias Exactas y Naturales, Universidad de Buenos Aires. Av.Cantilo s/n, Pabell\'on 1, Ciudad Universitaria, 1428, Buenos Aires, Argentina.}}
\affil[2]{\small{Instituto de F\'isica de Buenos Aires (IFIBA), CONICET. Av.Cantilo s/n, Pabell\'on 1, Ciudad Universitaria, 1428, Buenos Aires, Argentina.}}
\affil[3]{\small{Departamento de Matem\'atica, Facultad de Ciencias Exactas y Naturales, Universidad de Buenos Aires. Av.Cantilo s/n, Pabell\'on 1, Ciudad Universitaria, 1428, Buenos Aires, Argentina.}}

\theoremstyle{plain}

\newtheorem{thm}{Teorema}[section]
\newtheorem{prop}[thm]{Proposition}

\begin{document}
\maketitle

 \begin{abstract}
Usually, opinion formation models assume that individuals have an opinion about a given topic  which can change due to interactions with others. However, individuals can have different opinions in different topics 
and therefore n-dimensional models are best suited to deal with these cases. While there have been many efforts to develop analytical models for one dimensional opinion models, less attention has been paid to multidimensional ones.  In this work, we develop an analytical approach for multidimensional models of continuous opinions where dimensions can be correlated or uncorrelated. We show that for any generic reciprocal interactions between agents, the mean value of initial opinion distribution is conserved. Moreover, for positive social influence interaction mechanisms, the variance of opinion distributions decreases with time and the system converges to a delta distributed function. In particular, we calculate the convergence time when agents get closer in a discrete quantity after interacting, showing a clear difference between correlated and uncorrelated cases.

 \end{abstract}

\section{Introduction}

We commonly interact with other people either in person or through social media.  In these interactions an exchange of information can produce changes in subject's opinions. Although these processes are complex and heterogeneous, the elaboration of simplified models with statistical physic tools can help us to understand some of the observed phenomena through simple rules.

Opinion representation is a main ingredient in these models. Binary or discrete opinion spaces are used to represent an opinion as a set of options (for example a yes/no decision or a candidate for an election) \cite{Clifford,Holley,Liggett,Cox,Sire}. In continuous opinion models, when any intermediate value between two extremes is allowed, opinion is usually represented as a real number in a finite interval which represent the orientation between two given extremes.

Classical continuous models were grounded in different kinds of microscopic interactions. These models usually consider a discrete time, and at each time agents change their opinion influenced by others. DeGroot \cite{DeGroot} proposed one of the first models, where each agent change its opinion according to the weighted mean of other agents opinions. This is a deterministic, discrete, Markov chain which converges to a single consensus opinion if the transition matrix satisfies the irreducibly condition \cite{ross}. In \cite{D, D2} the authors proposed a model in which two given agents interacts if their opinion are closer than a threshold (bounded confidence hypothesis), and once they interact their opinions update proportionally to the difference between them. On the other hand, in  \cite{HK, Slanina} it is assumed that in each interaction an agent considers its surrounding  and take a mean of these  agents' opinions. Similar models without the proximity condition, or bounded confidence hypothesis,  are closer to the original model of DeGroot. These  models  were  extensively  studied  in last years \cite{Amblard-Deffuant, Deffuant-2006,Deffuant-2002, Lorenz-2008, Lorenz-2010}. They could be analyzed either with agent-based dynamics for a finite number  of $N$ interacting agents, or by using partial differential equations governing the evolution of a density function that represents the distribution of agents in the opinion space. These partial differential equations can be obtained as approximations of Boltzmann type equations modelling the agents interactions, as was shown in  \cite{Bell,PaTo, Tosc}.

An usual outcome in this kind of models is the formation of consensus. The consensus time in one dimensional models has been widely studied \cite{ANT,PPS, PPSS,PSB,Vazquez,Tyson,paperopinionFitness2}. However, there are few works studying this problem in higher dimensions.

In certain circumstances, a one dimensional opinion space is not enough since there are different topics or aspects of the opinion that are relevant. In political science, the traditional left–right distinction began to fail short to explain other ideologies. In 1971 David Nolan \cite{Nolan} created a political spectrum diagram charting political views along two axes, representing economic and personal freedom. Bryson and and McDill \cite{Bryson} proposed a similar diagram to classify political positions where coordinate on the vertical axis represents the degree of governmental control advocated (statism vs. anarchy), and position on the horizontal axis represents the degree of egalitarianism favored (left vs. right). This bi-dimensional classification lead to some quiz which enabled people to plot personal or political parties beliefs on a chart \cite{quiz1,quiz2}. This kind representations were also used in recent years by data-driven framework to classifies political opinion of society \cite{Krassa,FedeAlbanese,Falck}.

Multi-dimensional models are useful to represent and study political representation in this context. Axelrod classical model uses a multivariate sociological representation, even though is a cultural diffusion model \cite{Axelrod}. When opinions are represented in multidimensional topic space, correlation between components could play a key role in collective emergent states.
In \cite{Baumann}, the authors present a multidimensional model where collective state arise as a function of social influence and correlation between topics. This correlation is related to the influence that an opinion can exert on the opinion of another individual on another topic. The difference of correlation leads to difference in the steady state. Huet and Deffuant \cite{HD} also presented a bi-dimensional opinion model where they analyse the evolution using numerical simulations, observing that the cluster formation process is similar to the one dimensional model. Also in \cite{Laguna} the authors presents simulations for a multidimensional opinion model with for binary opinions. 

However, only few works formulated a general analytical framework for multidimensional opinions models. In \cite{Motsch_Tadmor} the authors prove the existence of consensus for an attractive multidimensional interaction. Noorazar \cite{Noorazar} raises a general variational model for opinion dynamics based on pairwise interaction potentials that comprise multidimensional models, but only get the master equation and study the dynamic for two and three agents with opinion in one dimension. In \cite{boudin} they found kinetic equations in higher dimensions for peer agents interactions (adding the role of media) and prove the existence and uniqueness of the solution, but they do not solve the dynamic.  
Even though the efforts observed in last years to analyze multidimensional opinion models, a complete analytical N-dimensional study  is still pending.

In this work, we propose a Boltzmann type equation for a general interaction model of opinion dynamics in a multidimensional space. The use of measure-valued solutions \cite{AzcS, PPSS} of these differential equations enable us to study precisely the behavior of the system with finitely many agents, as well as their limit when  $N \rightarrow \infty$, for small changes of opinion. Our main contributions can be summarized as follows:

i) We prove that for symmetrical interactions the mean value is preserved, and agents converge to this opinion if interactions are also attractive.

ii) We deduce an explicit analytical expressions for the evolution of variance of the opinion density function in the case where agents are attracted a finite quantity $h$ after interact in a n-dimensional opinion space. In particular, we compare the cases where the attraction is in the direction that shortens their distance (correlated opinion components) with the case where agents approach each other a distance $h$ independently in each coordinate. We show that the variance decreases quadratically in both cases.

iii) We obtain the consensus time, $t_c$ for the same specific interaction functions and found that the time of convergence is $8/5$ times faster in the uncorrelated model.

The paper is organized as follows. In the next section, we describe the rules of the models. In Section \ref{sec:ME} we obtain the kinetic equations for the agent distribution and we prove the main mathematical properties of the model related to  their mean values and the dispersion of agents and  we study the consensus time and we perform numerical simulations for correlated and uncorrelated models.

\section{Description of the model}

We consider a set of  $N$ agents where each one can have an opinion in several topics, and therefore its state is represented by an opinion vector in a $\R^n $ space. Opinions can only change  by interactions between pair of agents and are modeled as follows: at each time step $\Delta t = 1/N$, we select at random two agents in the population having opinions $\bm{x_i}$ and $\bm{x_j}$ and then update their opinions following the rule  described below:

\begin{equation}\label{UpDateRule}
\begin{array}{cc}
\bm{\tilde x_i}&=\bm{x_i}+h \bm{I}(\bm{x_j}-\bm{x_i}) \\
\bm{\tilde x_j}&=\bm{x_j}+h \bm{I}(\bm{x_i}-\bm{x_j}), 
\end{array}
\end{equation}
where  $\bm{\tilde x_i}$ and $\bm{\tilde x_j}$  denote their new opinion vectors,  $\bm{I}:\bm{z}\in\R^n \rightarrow \bm{I(z)}\in\R^n$ is a general interaction rule between two agents whose opinions differ in $\bm{z}$, and $h \in \R$ represents the opinion's shift. As can be seen, these equations  generalises a broad range of opinion's models. We suppose that $\bm{I}$ is an odd function, which corresponds to anti-symmetrical (both get identically attracted or repelled) interaction between any pair of agents. Moreover, if $\bm{z}\bm{I}(\bm{z)}<0$ it corresponds to attractive interactions, which has been widely studied  \cite{Abelson} either by mechanisms of persuasion \cite{Akers,Vinokur}, imitation \cite{Akers} or social pressure \cite{Asch,Homans,Sherif}.

This formulation also includes models where opinion's coordinates can be correlated, for instance, if the opinion update is produced in the absolute distance in the n-dimensional space. Opinion's shift are 
independent in each coordinate in the uncorrelated model. A more extensive analysis of this model will be describe in subsection \ref{time}.  

In figure (\ref{interaction}) we show some  examples of interaction functions in the $\R^2$ space. In panel (i) we show that an interaction function as $\bm{I}\bm{(z)}=\frac{\bm{z}}{\bm{|z|}}$ generates a symmetrical attraction in the direction between two given opinions. In panel (ii) we show similar interaction  with uncorrelated components: $I(\bm{z})=(\frac{z_x}{|z_x|},\frac{z_y}{|z_y|})$ in each opinion axis. In panel (iii) the interaction is similar to the first one, but repulsive:  $\bm{I}(\bm{z})=-\frac{\bm{z}}{\bm{|z|}}$, meanwhile in (iv) we show an example of interaction function where opinion update is proportional to the distance between interacting agents:  $\bm{I}(\bm{z})=\bm{z}$ 

\begin{figure}[htp] 
\includegraphics[width=0.9 \textwidth]{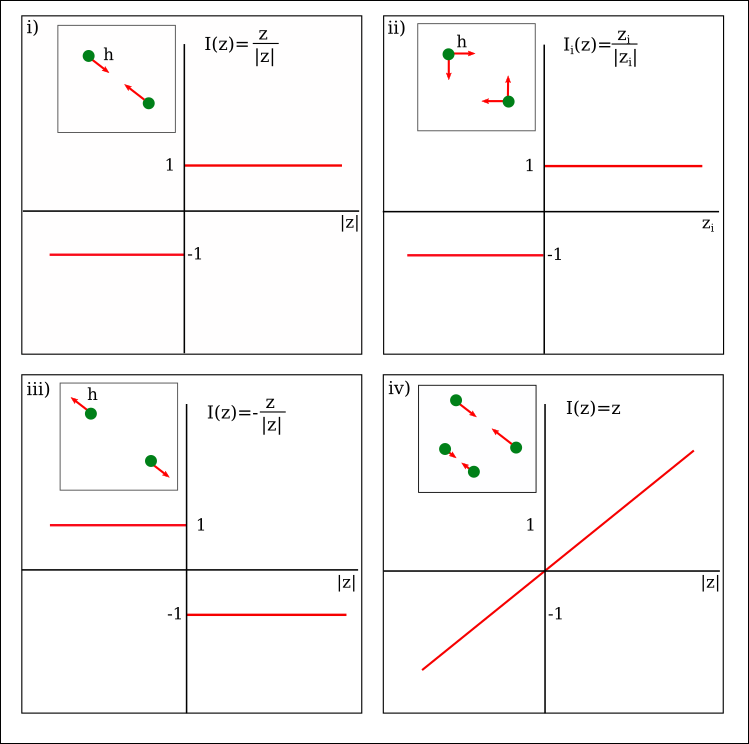}
\centering
\caption{\label{interaction} (i) $\bm{I}(\bm{z})=\frac{\bm{z}}{\bm{|z|}}$ generates a symmetrical attraction in the direction between two given opinions. (ii) Similar interaction  with uncorrelated components: $I(\bm{z})=(\frac{z_x}{|z_x|},\frac{z_y}{|z_y|})$. (iii) Similar to the first one, but repulsive:  $\bm{I}(\bm{z})=-\frac{\bm{z}}{\bm{|z|}}$ .(iv) The opinion update is proportional to the distance between interacting agents: $\bm{I}(\bm{z})=\bm{z}$. 
}
\end{figure}

\section{Master equations}\label{sec:ME}

To write down the master equations, we will follow the evolution of the opinion distribution in time, $f(t,\bm{x})$

The evolution of the opinion density function satisfies a master equation of the form
\begin{equation}\label{EcuacionMaestra}
\frac{\partial f}{\partial t}(t,\bm{x}) =2 \Big(G(\bm{x})-L(\bm{x})\Big),
\end{equation}
where $G(\bm{x})$ and $L(\bm{x})$ are the gain  and  loss functions respectively.  The $2$ factor represents the probability an agent changes its opinion at a time. 

Sometimes, it is easier to deal with the weak formulation of a given equation and also it allows that $f(t,\bm{x})$ be any type of measure-valued  not necessarily functions.
In Appendix \ref{ap:Weak Formulation} we deduce the weak formulation of Eq. \eqref{EcuacionMaestra} which reads as follow:

\begin{equation}\label{WeakForm}
\begin{split}
& \frac{d}{dt}\int_{\bm{x}} \phi(\bm{x})f(t,\bm{x})d\bm{x} \\ 
& = 2 \iint_{\bm{x},\bm{x_*}} 
[\phi(\bm{\bm{x}}+h I(\bm{\bm{x}_*}-\bm{x}))-\phi(\bm{x})]
f(t,\bm{x})f(t,\bm{x}_*) d\bm{x} d\bm{x}_*.
\end{split}
\end{equation}
for any function $\phi(\bm{x}): \R^n \rightarrow  \R^n$.
 
This expression is useful to obtain the dynamic for any observable values $\phi$.  For instance, if $\phi(\bm{x})=\bm{x}$, Eq. \eqref{WeakForm} gives the evolution of the geometric center of the opinion density function, and $\phi(\bm{x})=|\bm{x}|^2$ gives its the second moment.

In particular, Eq. \eqref{WeakForm} shows a simpler version in the limit $h\ll 1$:

\begin{equation}\label{WeakFormApprox}
\begin{split}
 \frac{d}{dt}\int_{\bm{x}} \phi(\bm{x})f(t,\bm{x})d\bm{x}
& =2h \iint_{\bm{x},\bm{x_*}}  \nabla\phi(\bm{x})I(\bm{\bm{x}_*}-\bm{x})f(t,\bm{x})f(t,\bm{x}_*)
d\bm{x} d\bm{x}_*.
\end{split}
\end{equation}

It will be useful for future sections to obtain the weak formulation for a time dependent function test  $\phi(t,\bm{x}): (\R,\R^n) \rightarrow  \R^n$ (See Appendix \ref{ap:Weak Formulation} for details). 

\begin{equation}\label{WeakFormApprox_t}
\begin{split}
 \frac{d}{dt}\int_{\bm{x}} \phi(t,\bm{x})f(t,\bm{x})d\bm{x}
& = \int_{\bm{x}} \frac{\partial\phi}{\partial t}(t,\bm{x}) f(t,\bm{x})d\bm{x}+\\
&2h \iint_{\bm{x},\bm{x_*}}  \nabla\phi(t,\bm{x})I(\bm{\bm{x}_*}-\bm{x})f(t,\bm{x})f(t,\bm{x}_*)
d\bm{x} d\bm{x}_*.
\end{split}
\end{equation}

\subsection{Evolution of first and second moment of the opinion density function}

In this section, we will focus in the asymptotic behavior of the opinion density function. In particular, we will prove that, under certain general conditions of the interaction function $I(\bm{z})$, an initial population of agents with randomly distribution opinion converge to consensus. This was done in one dimension when $I(z)=\frac{z}{|{z}|}$ in \cite{PSB}, where the authors proved that as $t\to +\infty$, $f(t,\cdot)$ converges (in a weak sense) to
the Dirac mass centered at the  mean value $m_0=\int {x}f(0,{x})d{x}$ of $f(0,\cdot)$: $f(t,\cdot) \to \delta_{m_0}$. In this section, we will generalize this result to a more general family of interaction functions and {\it n}-dimensions of the opinion vectors.

First, we will prove that, if $I(\bm{z})$ is an odd function, the mean opinion is conserved and that $f(t,\cdot)$ is symmetric if it was initially:

\medskip

\begin{prop}\label{mean_conseve}
If $I(-\bm{z})=-I(\bm{z})$  for any $\bm{z}\in\mathbb{R}^n$ (i.e. $\bm{I}$ is an odd function), then
\begin{equation}\label{MeanConserved}
  \int\bm{x} f(t,\bm{x})d\bm{x} = \int\bm{x} f(0,\bm{x})d\bm{x} \qquad \text{for any $t\ge 0$}.
\end{equation}

Moreover if $f(0,-\bm{x})=f(0,\bm{x})$ for any $\bm{x}\in\mathbb{R}^n$, then
$f(t,\cdot)$ is symmetric for any $t$.
\end{prop}

\begin{proof}
 Taking in \eqref{WeakForm}
the test function $\phi(\bm{x})=\bm{x}$ and using that $\nabla \phi(\bm{x})$ is the identity, we get
\begin{equation}
\begin{split}
\label{demSymetric1}
 \frac{d}{dt}\int \bm{x}f(t,\bm{x})d\bm{x}
& =2h \iint
I(\bm{\bm{x}_*}-\bm{x})
f(t,\bm{x})f(t,\bm{x}_*) d\bm{x} d\bm{x}_* \\
& =-2h \iint
I(\bm{\bm{x}}-\bm{x}_*)
f(t,\bm{x})f(t,\bm{x}_*) d\bm{x} d\bm{x}_* \\
& =-\frac{d}{dt}\int \bm{x}f(t,\bm{x})d\bm{x},
\end{split}
\end{equation}

where we used $I$ is odd and the change of variables $\bm{x}\to -\bm{x}$,
$\bm{x}_*\to -\bm{x}_*$ in the last equality. So $\frac{d}{dt}\int \bm{x}f(t,\bm{x})d\bm{x}
=0$ and the result follows.
We now suppose that $f(0,-\bm{x})=f(0,\bm{x})$ for any $\bm{x}\in\mathbb{R}^n$, and let $g(t,\bm{x}):=f(t,-\bm{x})$.
Then $\int \phi(\bm{x}) g(t,\bm{x})d\bm{x}
= \int \phi(\bm{-x}) f(t,\bm{x})d\bm{x}$
so that
\begin{equation}
\begin{split}
  \frac{d}{dt}\int \phi(\bm{x}) g(t,\bm{x})d\bm{x}
& =2h \iint
[\phi(-\bm{\bm{x}}-h I(\bm{\bm{x}_*}-\bm{x}))-\phi(-\bm{x})]
f(t,\bm{x})f(t,\bm{x}_*) d\bm{x} d\bm{x}_*.
\end{split}
\end{equation}

Changing variables $\bm{x}\to -\bm{x}$,
$\bm{x}_*\to -\bm{x}_*$, and then using that $I$ is odd again, we obtain
\begin{equation}
\begin{split}
\frac{d}{dt}\int \phi(\bm{x}) g(t,\bm{x})d\bm{x}
& =2h \iint
[\phi(\bm{\bm{x}}-h I(\bm{\bm{x}}-\bm{x}_*))-\phi(\bm{x})]
g(t,\bm{x})g(t,\bm{x}_*) d\bm{x} d\bm{x}_* \\
& =2h \iint
[\phi(\bm{\bm{x}}+h I(\bm{\bm{x}_*}-\bm{x}))-\phi(\bm{x})]
g(t,\bm{x})g(t,\bm{x}_*) d\bm{x} d\bm{x_*}.
\end{split}
\end{equation}

The transport Eq. \eqref{WeakFormApprox} is solved using the fixed point theorem of Banach, as in \cite{PPSS}, so the solution exist and is unique. Then, as $f(t,\bm{x})$ and $g(t,\bm{x})$ both solve the equation with the same initial condition, they are equal.
\end{proof}

 For attractive interactions we will investigate the long time behaviour of $f(t,\cdot)$ to see the opinion's convergence and generalize the results obtained in  \cite{PSB}.

\begin{prop}\label{Convergence2Dmodel2}
Assume that (i) $I$ is odd and also that
(ii) $\bm{z}I(\bm{z})\ge 0$ for all $\bm{z}\in \mathbb{R}^n$ with equality only for $\bm{z}=0$.
Then $f(t,\cdot)\to \delta_{\bm{m}(0)}$ where $\bm{m}(0)$ is the initial mean opinion.
\end{prop}

\begin{proof}
We assume, without loss of generality, that $f(0,\cdot)$ has  zero mean. It  follows from the previous result that $f(t,\cdot)$ has also mean zero for any $t\ge 0$.
It thus suffices to prove that
\begin{equation}\label{VarTo0}
\int |\bm{x}|^2 f(t,\bm{x})d\bm{x}\to 0
\qquad \text{as  $t \to +\infty$.}
\end{equation}
Taking $\phi(\bm{x})=\frac12 |\bm{x}|^2$ in \eqref{WeakFormApprox} gives
\begin{equation} 
 \frac{1}{2} \frac{d}{dt}\int  |\bm{x}|^2 f(t,\bm{x})d\bm{x}
 =2h\iint  \bm{x}I(\bm{\bm{x}_*}-\bm{x}) f(t,\bm{x})f(t,\bm{x}_*) d\bm{x} d\bm{x}_*.
 \end{equation}

Let us call $A=\iint  \bm{x}I(\bm{\bm{x}_*}-\bm{x}) f(t,\bm{x})f(t,\bm{x}_*) d\bm{x} d\bm{x}_*$, so that \begin{equation}\label{2A}
\frac{1}{2h} \frac{d}{dt}\int  |\bm{x}|^2 f(t,\bm{x})d\bm{x}
 = 2A
\end{equation}
Using that $I$ is odd we can also write $A$ as
\begin{equation}\label{demECM1}
\begin{split}
A & = -\iint
\bm{x}I(\bm{\bm{x}}-\bm{x}_*)
f(t,\bm{x})f(t,\bm{x}_*) d\bm{x} d\bm{x}_* \\
& =  -\iint
\bm{x}_*I(\bm{\bm{x}_*}-\bm{x})
f(t,\bm{x})f(t,\bm{x}_*) d\bm{x} d\bm{x}_*,
\end{split}
\end{equation}
where we performed the change of variables $\bm{x}_* \rightarrow \bm{x} $, $\bm{x} \rightarrow \bm{x}_* $ in the last equality. 
Using this last expression and the definition of $A$, we obtain
\begin{equation}
2A = -\iint
(\bm{x}_*-\bm{x})I(\bm{\bm{x}_*}-\bm{x})
f(t,\bm{x})f(t,\bm{x}_*) d\bm{x} d\bm{x}_*.
\end{equation}

We can thus rewrite  Eq. \eqref{2A} as
\begin{equation}\label{Var2}
\frac{1}{2h}\frac{d}{dt}
\int  |\bm{x}|^2 f(t,\bm{x})d\bm{x}
  = -\iint
(\bm{x}_*-\bm{x})I(\bm{\bm{x}_*}-\bm{x})
f(t,\bm{x})f(t,\bm{x}_*) d\bm{x} d\bm{x}_*,
\end{equation}
which is always negative by assumption (ii).
Moreover, still by (ii),
\begin{align*}\label{demECM2}
A=0
& \Longleftrightarrow &
 (\bm{x}_*-\bm{x})I(\bm{\bm{x}_*}-\bm{x})
 \qquad \text{for any $\bm{x},\bm{x}_*\in \text{supp}(f(t,\cdot))$} \\
 & \Longleftrightarrow &
 \bm{x}_*=\bm{x}
 \qquad \text{for any $\bm{x},\bm{x}_*\in \text{supp}(f(t,\cdot))$} \\
 & \Longleftrightarrow &
 \text{$f(t,\cdot)$ is  a Dirac mass.}
\end{align*}
Thus the variance $\int  |\bm{x}|^2 f(t,\bm{x})d\bm{x}$ of $f(t,.)$
is non-increasing in time, and is decreasing strictly while $f(t,\cdot)$ is not a Dirac mass.

Since the mean value of $f(t,\cdot)$ is conserved, and thus is equal to $\bm{m}(0)$,
$f(t,\cdot)$ must converge $\delta_{\bm{m}(0)}$.
\end{proof}

In figure (\ref{snapshot}) we show a simulation for $N=1000$ agents in a 2D opinion space  starting from an uniform distribution  $[-0.5; 0.5]\times [-0.5;0.5]$ and space step $h=0.01$. We show snapshots at different times of evolution for two different interaction functions. 

\begin{figure}[ht]
 \includegraphics[width=\textwidth]{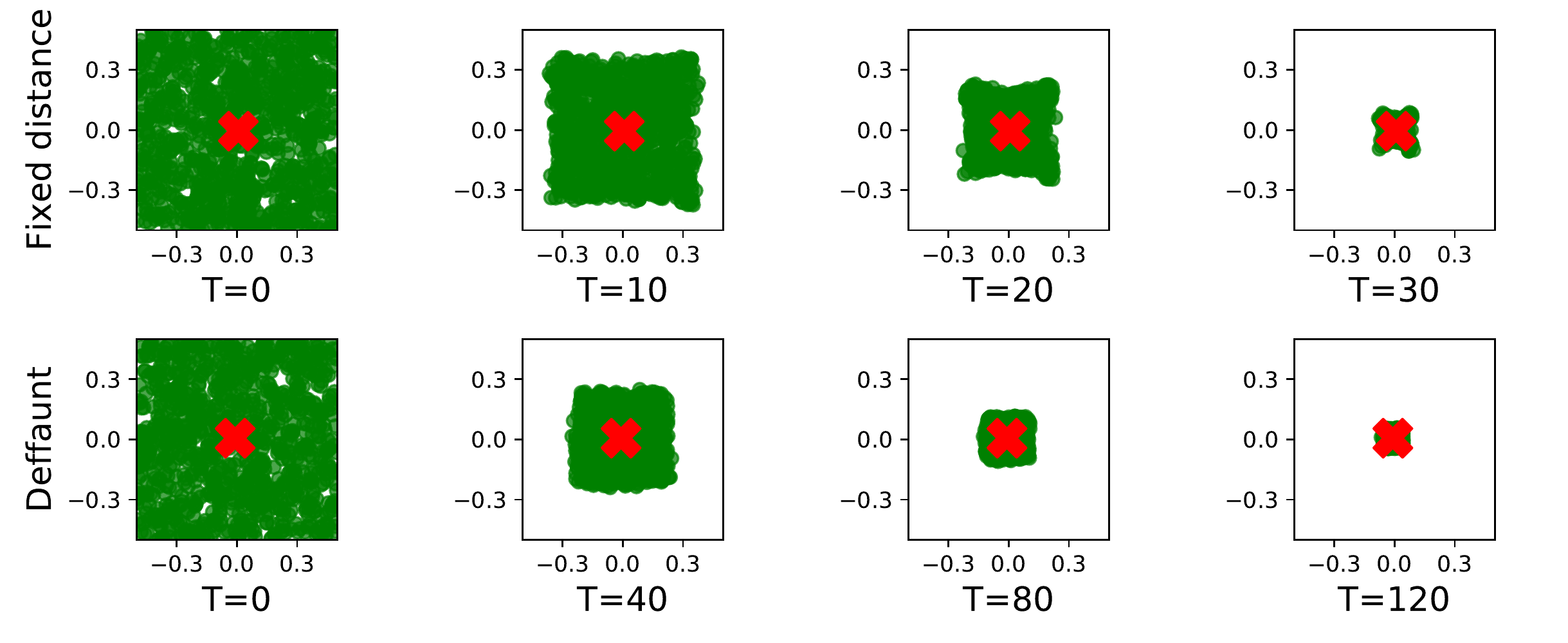}
\caption{ System evolution simulation for $N=1000$ agents, $h=0.01$ and initially uniform distribution in a 2D opinion space. Top panel use  interaction function given by $\bm{I}\bm{(z)}=\frac{\bm{z}}{\bm{|z|}}$ (generates a symmetrical attraction in the direction between two given opinions) and bottom panel $\bm{I}(\bm{z})=\bm{z}$ (opinion update is proportional to the distance between interacting agents). Red cross represent the mean opinion at each time.}
\label{snapshot}
\end{figure}

In top panels of figure (\ref{snapshot}), the interaction function is $\bm{I}\bm{(z)}=\frac{\bm{z}}{\bm{|z|}}$ which generates a symmetrical attraction in the direction between two given opinions. In bottom panels the interaction function is $\bm{I}(\bm{z})=\bm{z}$, where opinion update is proportional to the distance between interacting agents. We can see that both interaction functions are odd an fulfill the condition $\bm{z}I(\bm{z})\ge 0$ for all $\bm{z}\in \mathbb{R}^2$ with equality only for $\bm{z}=0$, thus making the distribution of opinions converge toward the mean opinion (marked with a red cross) in agreement with Proposition \ref{Convergence2Dmodel2}.

\subsection{Time to consensus }\label{time}

In this section, we deduce an analytical expression for the evolution of the second moment of the opinion density function which will also allow us to obtain  the time needed by the collective of agents to reach consensus, $t_c$.
We focus on a particular interaction function, $\bm{I(z)}=\frac{\bm{z}}{|\bm{z}|}$,  which produces that every time two agents interact, they approach a distance $h$ in the direction that shortens their distance (note that in this case, both opinion coordinates are correlated). We will compare the obtained results with the case where opinion coordinates update in a uncorrelated way: $I(\bm{z})=(\frac{z_x}{|z_x|},\frac{z_y}{|z_y|})$.  For sake of simplicity we use $d=2$ but results could be extended to more dimensions.

It follows from the previous analysis that the dynamics of correlated and uncorrelated interactions considered here lead to consensus with the same limit opinion, namely the initial mean opinion. We will show however that the time to consensus is different in both models.

We can assume (without lose of generality) that $f(t,\cdot)$ has  zero mean if initial opinions are symmetrically distributed around $x=0$. Thus, according to Proposition
\ref{mean_conseve} proved in the previous section, $f(t,\cdot)$ has mean zero for any $t$.
The variance of $f(t,.)$ then simplifies to 
$\langle \bm{x}^2 \rangle$  given by
\begin{equation}
  \langle \bm{x}^2 \rangle(t) = \int  |\bm{x}|^2 f(t,\bm{x})d\bm{x}. \end{equation}

Recalling that $\bm{I(z)}=\frac{\bm{z}}{|\bm{z}|}$, $(\bm{x}_*-\bm{x})I(\bm{x}_*-\bm{x})=|\bm{x}_*-\bm{x}|$ and Eq. \eqref{Var2} gives

\begin{equation}\label{Var3}
\frac{1}{2h} \frac{d \langle \bm{x}^2 \rangle(t)}{dt}
  = -\iint |\bm{x}_*-\bm{x}| f(t,\bm{x})f(t,\bm{x}_*) d\bm{x} d\bm{x}_*.
\end{equation}

We note that $\iint |\bm{x}_*-\bm{x}| f(t,\bm{x})f(t,\bm{x}_*) d\bm{x} d\bm{x}_*=\text{MD}^{L_2}_{\bm{x}}(t)$ is the mean absolute difference with the $L_2$ distance, so that

\begin{equation}\label{Var3_MD}
\frac{d \langle \bm{x}^2 \rangle(t)}{dt}
  = -2h \text{MD}^{L_2}_{\bm{x}}(t).
\end{equation}

In order to derive a second times this expression, we let $\Phi(t,\bm{x})=\int |\bm{x}_*-\bm{x} |f(t,\bm{x}_*)d\bm{x}_*$ so that
\begin{equation}\label{Var4}
\frac{1}{2h} \frac{d \langle \bm{x}^2 \rangle(t)}{dt}
  = -\int \Phi(t,\bm{x})  f(t,\bm{x}) d\bm{x}. 
\end{equation}
We take the second time derivative using  Eq. \eqref{WeakFormApprox_t} expression with the  time-dependent function $\Phi(t,x)$. We obtain  \begin{equation}\label{Var2_bis}
\begin{split}
-\frac{1}{2h} \frac{d^2\langle \bm{x}^2 \rangle(t)}{dt^2}=&
  \int \frac{\partial \Phi}{\partial t}(t,\bm{x})  f(t,\bm{x}) d\bm{x} \\ 
  & +2h \iint \nabla \Phi(t,\bm{x}) I(\bm{\bar{x}}-\bm{x}) f(t,\bm{x})f(t,\bm{\bar{x}})d\bm{x}d\bm{\bar{x}}.
\end{split}
\end{equation}
 To compute $\frac{\partial \Phi}{\partial t}(t,\bm{x})$, we use  Eq. \eqref{WeakFormApprox} with the test function $\phi(\bm{x}_*)=|\bm{x}_*-\bm{x}|$. Using that  
 $\nabla\phi(\bm{x}_*)=(\bm{x}_*-\bm{x})/|\bm{x}-\bm{x}_*|$, we obtain
\begin{equation}
\begin{split}
\frac{\partial\Phi}{\partial t}(t,\bm{x})
 & = \frac{\partial}{\partial t} 
 \int |\bm{x}_*-\bm{x}| f(t,\bm{x}_*)d\bm{x}_* \\ 
& =  2h\iint \frac{(\bm{x}_*-\bm{x})}{|\bm{x}-\bm{x}_*|} I(\bm{\bar{x}}-\bm{x}_*)f(t,\bm{\bar x}) f(t,\bm{x}_*)d\bm{\bar x} d\bm{x}_*.
\end{split}
\end{equation}

back to \eqref{Var2_bis} with  $\nabla \Phi(t,\bm{x})=\int \frac{\bm{x}-\bm{x}_*}{|\bm{x}-\bm{x}_*|}f(t,\bm{x}_*)d\bm{x}_*$ and 
$I(\bm{z})=\bm{z}/|\bm{z}|$, we get 

\begin{equation}
\begin{split}
 \frac{1}{4h^2} \frac{d^2\langle \bm{x}^2 \rangle(t)}{dt^2}
& =  \iiint \Big( \frac{\bm{x}-\bm{x}_*}{|\bm{x}-\bm{x}_*|}\frac{\bm{\bar{x}}-\bm{x}_*}{|\bm{\bar{x}}-\bm{x}_*|}+\frac{\bm{x}_*-\bm{x}}{|\bm{x}_*-\bm{x}|}\frac{\bm{\bar{x}}-\bm{x}}{|\bm{\bar{x}}-\bm{x}|}\Big)f(t,\bm{ x})f(t,\bm{\bar x}) f(t,\bm{x}_*)d\bm{x}d\bm{\bar x} d\bm{x}_*  \\ 
& =  2\iiint \Big( \frac{\bm{x}_*-\bm{x}}{|\bm{x}_*-\bm{x}|}\frac{\bm{\bar{x}}-\bm{x}}{|\bm{\bar{x}}-\bm{x}|}\Big)f(t,\bm{ x})f(t,\bm{\bar x}) f(t,\bm{x}_*)d\bm{x}d\bm{\bar x} d\bm{x}_*.
\end{split}
\end{equation}

Eventually for symmetry reasons, 
\begin{equation}\label{Der2_bis}
\begin{split}
\frac{3}{8h^2}\frac{d^2\langle \bm{x}^2 \rangle(t)}{dt^2}
&= \iiint \Big( \frac{\bm{x_*}-\bm{x}}{|\bm{x_*}-\bm{x}|}
                \frac{\bm{\bar x}-\bm{x}}{|\bm{\bar x}-\bm{x}|}
+ \frac{\bm{\bar {x}}-\bm{x}_*}{|\bm{\bar {x}}-\bm{x}_*|} \frac{\bm{x}-\bm{x}_*}{|\bm{x}-\bm{x}_*|} \\
& \hspace{1cm} + \frac{\bm{x}_*-\bm{\bar x}}{|\bm{x}_*-\bm{\bar x}|} \frac{\bm{x}-\bm{\bar x}}{|\bm{x}-\bm{\bar x}|} \Big)
f(t,\bm{x})f(t,\bm{x}_*)f(t,\bm{\bar x})
d\bm{\bar x} d\bm{x} d\bm{x}_*.
\end{split}
\end{equation}

Notice that given three different points $\bm{x},\bm{x}_*,\bm{\bar{x}}\in \mathbb{R}^2$, (which form a triangle in the plane) the addition in the integral is the sum of the cosines of the angles of this triangle.
Since those angles must be between 0 and $\pi/2$ and sum $\pi$,
 the sum of the cosines is maximum equal to $3/2$ for an equilateral triangle, and is minimum equals to 1 when one angle is 0 and the others are $\pi/2$. 

Notice also that the probability of not choosing 3 distincts points among $N$ is negligible neing  of order $1/N$. 
We thus deduce from \eqref{Der2_bis} that 

 \begin{equation}\label{Der2R2}
  \frac83 h^2
 \le \frac{d^2\langle \bm{x}^2 \rangle(t)}{dt^2} 
 \le  4h^2.
\end{equation}

Integrating w.r.t $t$, 
using Eq. \eqref{Var3_MD} to calculate $\frac{d\langle \bm{x}^2 \rangle}{dt}(0)$, we obtain 

\begin{equation}\label{ECM_2D}
\frac43 h^2 t^2 - 2 h \text{MD}^{L_2}_{\bm{x}}(0) t  +\langle \bm{x}^2 \rangle(0)
\le \langle \bm{x}^2 \rangle(t)  \le 2 h^2 t^2 - 2h  \text{MD}^{L_2}_{\bm{x}}(0) t  +\langle \bm{x}^2 \rangle(0).
\end{equation}

If we assume than $\langle \bm{x}^2 \rangle(t)$ is the mean between the lower and upper bound : 

\begin{equation}\label{ECM_2D_2}
\langle \bm{x}^2 \rangle(t) \approx \frac{5}{3} h^2 t^2 -2h \text{MD}^{L_2}_{\bm{x}}(0) t +\langle \bm{x}^2 \rangle(0).
\end{equation}
This assumptions is reasonable if we consider that we are averaging over many agents and therefore,  three selected agents should have angles probabilities equally distributed in the opinion space.

For the uncorrelated case we use the result derived in Appendix \ref{ap:Uncorrelated model} which shows that the dynamic is equivalent to the 1D case in each coordinate,  and therefore the consensus time is equivalent to the 1D model. For this model the master equation can be written as 
Eq.\eqref{Uncorrelated_WF_ap_2} (See Appendix \ref{ap:Uncorrelated model}) taking the function $\phi(\bm{x})=\frac12 \bm{x}^2=\frac12 ( x^2+y^2)$. 
We get
\begin{equation}\label{Var_2dim}
\begin{split}
\frac{1}{2h} \frac{d\langle \bm{x}^2 \rangle(t)}{dt}&=- \iint (|x_{*}-x|+|y_{*}-y|)f(t,\bm{x}_*)f(t,\bm{x})d \bm{x}_* \bm{x} \\
&= -\text{MD}_{x}(t)-\text{MD}_{y}(t)=-2\text{MD}_{x}(t)=-\text{MD}^{L_1}_{\bm{x}}(t),
\end{split}
\end{equation}
where $\text{MD}_{x}(t)=\text{MD}_{y}(t)$ is the mean absolute difference for the projection to one dimension, and $\text{MD}^{L_1}_{\bm{x}}(t)$ is the mean absolute difference with the $L_1$ distance. 

Using a deduction similar to the correlated case but considering each coordinate as the 1D model, we get 
\begin{equation}\label{Der2_bis_x_y}
\begin{split}
\frac{3}{8 h^2}\frac{d^2\langle {x}^2 \rangle(t)}{dt^2}
&= 2\iiint_{\R^3} \Big( \frac{{x_*}-{x}}{|{x_*}-{x}|}
   \frac{{\bar x}-{x}}{|{\bar x}-{x}|}
+ \frac{{\bar x}-{x}_*}{|{\bar x}-{x}_*|} \frac{{x}-{x}_*}{|{x}-{x}_*|} \\
& \hspace{1cm} + \frac{{x}_*-{\bar x}}{|{x}_*-{\bar x}|} \frac{{x}-{\bar  x}}{|{x}-{\bar x}|} \Big)
f_t({x})f_t({x}_*)f_t({\bar x})
d{\bar x} d{x} d{x}_*.
\end{split}
\end{equation}
Even though  there are no triangle in the real axis as in the previous deduction, we can think 
of 3 distinct points ${x},{x_*},{\bar{x}}\in \R$ as  forming 
a degeneate triangle with one angle equals to $\pi$ and the others two equal to zero,  making the sum of their cosines equals to one. Thus, 
\begin{equation}\label{Der2R1}
 \frac{d^2\langle \bm{x}^2 \rangle(t)}{dt^2} = \frac{16}{3} h^2. 
\end{equation}
Integrating twice yields
\begin{equation}\label{ECM_1D}
\langle \bm{x}^2 \rangle(t) = \frac{8}{3} h^2 t^2 - 2 h  \text{MD}^{L_1}_{\bm{x}}(0) t+\langle \bm{x}^2 \rangle(0). 
\end{equation}

The consensus time can be calculated analytically from the minimum of Eqs. \eqref{ECM_2D_2} and \eqref{ECM_1D}, resulting in the following expressions: 

\begin{equation}\label{ConvergenceTime}
t_c =
\left\{  \begin{array}{ll}
 \displaystyle {\frac{ 3 \text{MD}^{L_1}_{\bm{x}}(0)}{8h}} & \quad  \text{for the uncorrelated model, }\\
  \displaystyle \frac{  3\text{MD}^{L_2}_{\bm{x}}(0)}{ 5h} & \quad  \text{for the correlated model. }
\end{array}\right.
\end{equation}
\begin{figure}[ht]
\centering
\includegraphics[width=\textwidth]{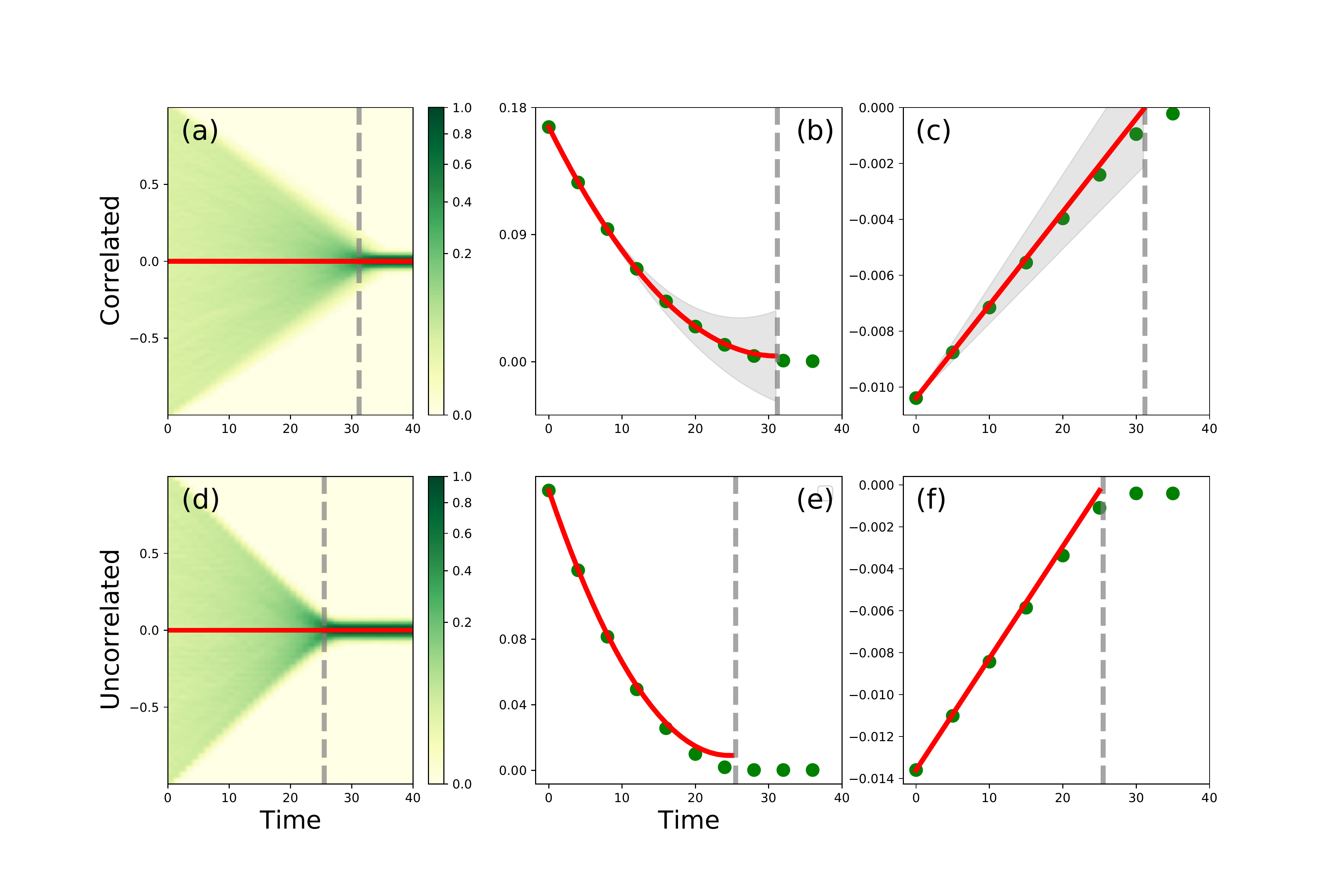}
\caption{ System evolution simulation for $N=1000$ agents, $h=0,01$ and initially uniform distribution. Panels (a) and (d) show the projection of the density opinion function onto the x-axis. Middle panels ((b) and (e)) display simulations (green circles) and theoretical  estimation (red lines) of $\langle \bm{x}^2 \rangle(t)$. In right panels ((c) and (f)) we can see simulations (green circles) and theoretical  estimation (red lines) of $\text{MD}^{L_2}_{\bm{x}}(t)$ and $\text{MD}^{L_1}_{\bm{x}}(t)$ respectively (see the text for the details of the computations). Notice than top panels show the behavior of the correlated model meanwhile bottom ones the results corresponding to the  uncorrelated model.
}
\label{density_MSE}
\end{figure}

In panels (a) and (d) of figure (\ref{density_MSE}) we show  the opinion  density function projected to the $x$ coordinate (equivalent results are obtained in the $y$ coordinate), and we can see the difference for time convergence. In panels (b) and (e) we show  the time evolution of
$\langle \bm{x}^2 \rangle(t)$ calculated (i) from  the simulation according to $\frac{1}{N}\sum_{i=1}^N x^i(t)$ (green circles), and (ii) from Eqs. \eqref{ECM_2D_2} and \eqref{ECM_1D} (red lines). In panels (c) and (f) we show the time evolution of  $\text{MD}^{L_2}_{\bm{x}}(t)$ and $\text{MD}^{L_1}_{\bm{x}}(t)$ calculated from  the simulation according to $\frac{1}{N^2}\sum_{i=1,j=1}^{N,N} |\bm{x}^i(t)- \bm{x}^j(t)|$ and $\frac{1}{N^2}\sum_{i=1,j=1}^{N,N} |x^i(t)- x^j(t)|+|y^i(t)- y^j(t)|$ respectively (green circles), and (ii) derivating Eqs. \eqref{ECM_2D_2} and \eqref{ECM_1D} (red lines). Analytical  consensus time $t_c$ are plotted in dash gray lines. We observe an excellent agreement.

Notice that $\text{MD}^{L_1}_{\bm{x}}(0)$ and $\text{MD}^{L_2}_{\bm{x}}(0)$ are different for both models. 
Indeed it follows from  $|.|_{L^1}\geq |.|_{L^2} \geq \frac{1}{\sqrt2} |.|_{L^1}$ that  
\begin{equation}\text{MD}^{L_1}_{\bm{x}}(t) \geq \text{MD}^{L_2}_{\bm{x}}(t) \geq \frac{1}{\sqrt2} \text{MD}^{L_1}_{\bm{x}}(t).  
\end{equation}
 
This is also valid for $t=0$ and therefore  allow us to compare both convergence times. For instance, in the simulations performed in figures (\ref{density_MSE}) and (\ref{ConvergenceTime_1D_vs_2D_SimulTeo_Normalized}) where agents start from an uniform distribution in $[-0.5; 0.5]\times [-0.5;0.5]$ we get, independently from $h$ and $N$, $\text{MD}^{L_1}_{\bm{x}}(0) \approx 0.66$ and $\text{MD}^{L_2}_{\bm{x}}(0) \approx 0.52$, making the uncorrelated model reach the consensus faster.

Consensus time $t_c$ can be also estimated numerically  from the agent-based simulation when:
\begin{equation}\label{tc_Simul}
|\langle \bm{x}^2 \rangle(t)-\langle \bm{x}^2 \rangle(t-1)| =
 \Big|\frac1N \sum_{i=1}^N |\bm{x^i}(t)|^2 - \frac1N \sum_{i=1}^N |\bm{x^i}(t-1)|^2 \Big|
< h^2.
\end{equation}

\begin{figure}[ht]
\includegraphics[width=\textwidth]{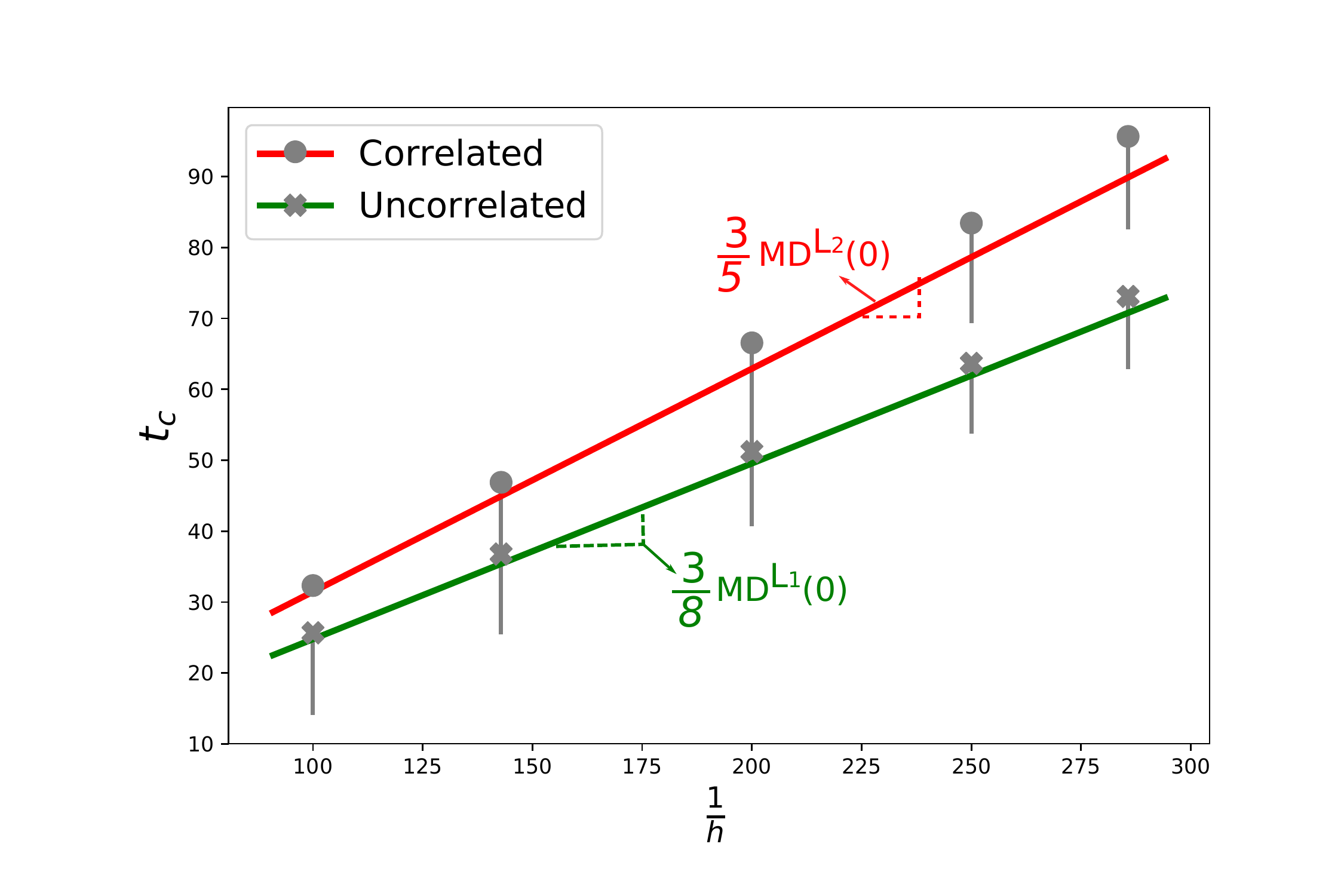}
\caption{\label{ConvergenceTime_1D_vs_2D_SimulTeo_Normalized} Proportionality between convergence time over $\langle \bm{x}^2 \rangle$ and $\frac{1}{h}$ for both correlated and uncorrelated models. Analytical results are in lines and simulation in symbols. Simulations use an initial uniform distribution and varies $N=1000, 1500, 2000$   }
\end{figure}

Given that the time to consensus, $t_c$,  depends only of $h$ and the initial distribution, in figure (\ref{ConvergenceTime_1D_vs_2D_SimulTeo_Normalized}) we plot (with lines) the analytical expression for $tc$ as a function of $\frac{1}{h}$. As expected, a linear relation with slopes $\frac{3}{8}\text{MD}^{L_1}_{\bm{x}}(0) $ for the uncorrelated case and to $\frac{3}{5}\text{MD}^{L_2}_{\bm{x}}(0) $ for the correlated one are obtained. We compare the analytical result with agent-based simulation (with symbols). We simulate for different $N=1000, 1500, 2000$ and initial uniform distribution obtain identical results.  We can observe here a slight mismatch between simulations and analytical results. This happens because our analytical approach fails when the opinions distribution become a Dirac delta and therefore the simulation consensus time is overestimated. In order to get a lower bound for the consensus time obtained from simulations,  we estimate the time when the second derivative is not longer constant (roughly, when it felt apart about $h^2$ for the theoretical value), i.e.,
\begin{equation}\label{lower_bound}
\frac{d \text{MD}_{\bm{x}}(t) }{dt}=\frac{d }{dt}\sum_{i=1,j=1}^{N,N} |\bm{x^i}(t)-\bm{x^j}(t)| <t_{c}-h^2.
\end{equation}  
In this way, we can observe an excellent agreement between analytical results and simulations within the error bars.

\section{Discussions}

In this work we present a general framework to study analytically continuous opinion models with general interactions functions in a multidimensional space.  In this context we study the mean and the variance of the opinion distribution of a set of agents along a time and find geometrical differences between correlated and uncorrelated models that lead to differences in the dynamic and in the convergences times.

We deduce the master equation in the general case and prove that for symmetrical interactions the mean of the opinion of the agents remains constant, and if the interactions are also attractive, the variance goes to $0$ and agent's opinions converges to consensus.

Additionally, we get an explicit expression for the variance of the opinion distribution, and we show that it depends quadratically on time, allowing the  calculation of the convergence time for two particular models: a first one when interacting agents approach a distance $h$ in the direction that shortens their distance (correlated opinion components) and a second one where  agent's approach each other a distance $h$ in each coordinate. 

The ratio of the convergence times $t_c^c/t_c^u$ (supra index c for correlated, u for uncorrelated) reads like this:
\begin{equation} 
\frac{t_c^c}{t_c^u} =\frac{8}{5}\frac{\text{MD}^{L_2}_{\bm{x}}}{\text{MD}^{L_1}_{\bm{x}}}.
\end{equation}

On one hand, the factor $\frac{\text{MD}^{L_2}_{\bm{x}}}{\text{MD}^{L_1}_{\bm{x}}}$ is bounded in $[\frac{1}{\sqrt2},1]$. We interpret this result with the difference in opinion shift of agents after interactions, since for the correlated case each agent close its partner a distance $h$, meanwhile it is $\sqrt{2}h$ in the uncorrelated one. On the other hand, the $8/5$ factor can be seen as a geometric phenomenon that comes from the difference between the quadratic term of expressions \eqref{ECM_2D_2} and \eqref{ECM_1D}. For these calculations the sum of the cosines of all the angles formed by three different agents are involved. Meanwhile in the correlated cases it involves triangles of opinions in $2D$ space, this triangles are confined to  one dimension in the uncorrelated case, as was explained above.

Even though a faster convergence of uncorrelated model was expected, given that agents moves $\sqrt2 h$ at each interaction in comparison with the correlated one where they move $h$, we notice a non-trivial relation between both convergence times, as was analyzed in previous paragraph.

\section*{Acknowledgement}

This work was supported by Grants No. 20020130100582BA   (UBACyT) and No. PICT-201-0215(ANPCyT), Argentina.

\begin{appendices}

\section{Weak Formulation}
\label{ap:Weak Formulation}

In this section we write the master equation \eqref{EcuacionMaestra} in the weak form. 

To get a $\bm{x}$ position we have to choose two agents where interaction makes one of them adopt the $\bm{x}$ so is the probability of interaction of two previous agents  $\bm{{x^{p}}}$ and an $\bm{{x_*^{p}}}$ when new opinion for first agent is $\bm {x}$. 

\begin{equation}\label{gain}
G(\bm{x})=\int_{x_*^{p}} P_{(\bm{x^{p}},\bm{x_*^{p}})\rightarrow {(\bm{x},\bm{x_*})}}f(\bm{x^{p}})f(\bm{x_*^{p}})d\bm{x_*^{p}}.
\end{equation}

Where $P_{(\bm{x^{p}},\bm{x_*^{p}})\rightarrow {(\bm{x},\bm{x_*})}}$ is the probability that $(\bm{x^{p}},\bm{x_*^{p}})$ get $(\bm{x},\bm{x_*})$ after an interaction. Allowing all the agents interact between them with same probability, $P$ is only $1$ when $\bm{x}=\bm{x^p}+I(\bm{x_*^p}-\bm{x^p})$.

To left a position we have to choose an agent of opinion $\bm{x}$ interacts whit some one who moves it. 
We get 

\begin{equation}\label{lost}
L(\bm{x})=\int_{x_*}   P_{(\bm{x},\bm{x_*})\rightarrow {(\bm{\tilde{x}},\bm{\tilde{x}_*})}}f(\bm{x})f(\bm{x_*})d\bm{x_*}.
\end{equation}

Given a test function $\phi:\R^n \rightarrow \R$ and integrating in the opinion space, we get 

\begin{equation}\label{Weak_Form_ap_1}
 \begin{split}
 \frac{d}{dt}\int \phi(x)f(t,\bm{x})d\bm{x} = &  
 2 \Big(\iint_{x_*^{p},x} \phi(\bm{x})  P_{(\bm{x^{p}},\bm{x_*^{p}})\rightarrow {(\bm{x},\bm{x_*})}}f(\bm{x^{p}})f(\bm{x_*^{p}})d\bm{x_*^{p}} 
d\bm{x}\\
 &-\iint_{x_*,x}\phi(\bm{x})   P_{(\bm{x},\bm{x_*})\rightarrow {(\bm{\tilde{x}},\bm{\tilde{x}_*})}}f(\bm{x})f(\bm{x_*})d\bm{x}d\bm{x_*}\Big).
 \end{split}
\end{equation}

By using that $P_{(\bm{x^{p}},\bm{x_*^{p}})\rightarrow {(\bm{x},\bm{x_*})}}=\delta_{\bm{x}=\bm{x^p}+I(\bm{x_*^p}-\bm{x^p})}$,

\begin{equation}
\begin{split} 
 \label{Weak_Form_ap_2}
 \frac{d}{dt}\int \phi(x)f(t,\bm{x})d\bm{x}
=  
& 2 \iint_{x_*^{p},x^{p}} \phi(\bm{x^p}+I(\bm{x_*^p}-\bm{x^p})  f(\bm{x^{p}})f(\bm{x_*^{p}})d\bm{x_*^{p}}d\bm{x^{p}} 
\\
&-\iint_{x^*,x}\phi(\bm{x})   f(\bm{x})f(\bm{x_*})d\bm{x}d\bm{x_*},
\end{split}
\end{equation} and 
by changing variables from $(\bm{x_*^{p}},\bm{x^p})$ to $(\bm{x_*},\bm{x})$ in gain term we get \eqref{WeakForm}.

If test function is time depending, $\phi(t,\bm{x}):(\R,\R^n)\rightarrow \R$ then

\begin{equation}
\begin{split}
\label{Weak_Form_ap_3}
 \frac{d}{dt}\int \phi(t,\bm{x})f(t,\bm{x})d\bm{x}
=  \int \frac{d \phi(t,\bm{x})}{dt} f(t,\bm{x}) d \bm{x} + \int \phi(t,\bm{x}) \frac{d f(t,\bm{x})}{dt} d\bm{x},
\end{split}
\end{equation}

At second term it can be used the eq. \eqref{WeakFormApprox} and we get \eqref{WeakFormApprox_t}

\section{Uncorrelated model}
\label{ap:Uncorrelated model}

In this appendix we will show that uncorrelated model is equivalent to a 1D model en each coordinate.

 Using the updating rule uncorrelated in \eqref{WeakForm} assuming $\bm{x}=(x,y)$ we obtain
\begin{equation}\label{Uncorrelated_WF_ap}
 \begin{split}
& \frac{d}{dt}\int \phi(\bm{x})f(t,\bm{x})d\bm{x}
=  \\
&  = \iint [\phi(x+hI(x^*-x),y+hI(y^*-y))-\phi(\bm{x})]f(t,\bm{x})f(t,\bm{x_*})
d\bm{x}d\bm{x_*}
\end{split}
\end{equation}

which can be approximated when $h\ll 1$ by
\begin{equation}\label{Uncorrelated_WF_ap_2}
\begin{split}
& \frac{1}{h} \frac{d}{dt}\int \phi(\bm{x})f(t,\bm{x})d\bm{x} \\
 & =  \iint \Big(\frac{\p\phi}{\p x}(\bm{x})I(x^*-x)+\frac{\p\phi}{\p y}(\bm{x})I(y^*-y)
	\Big)f(t,\bm{x})f(t,\bm{x_*})d\bm{x}d\bm{x_*}.
\end{split}
\end{equation}
Consider the vector field $\bm{v}=(v_x,v_y)$ defined by
\begin{equation}
\begin{split}
v_x(\bm{x})=v[f](x)=\int I(x^*-x) f(t,\bm{x_*})d\bm{x_*} \\
v_y(\bm{x})=v[f](y)=\int I(y^*-y) f(t,\bm{x_*})d\bm{x_*} 
\end{split}
\end{equation} 

To simplify this expression we introduce the marginals $F_x(t,x)$ and $F_y(t,y)$  of $f(t,\bm{x})$, namely
\begin{equation}\label{DefMarginals}
 F_x(t,x) =  \int^{\infty}_{-\infty}f(t,x,y) )dy, \qquad
 F_y(t,y) =  \int^{\infty}_{-\infty} f(t,x,y)dx.
\end{equation}
Notice that $F_x(t,x)$ and   $F_y(t,y)$ are probability measures on $\mathbb{R}$.
Indeed they are the projection of  $f(t,\bm{x})$ on the coordinate axis.
Then
\begin{equation}\label{VF2}
 v_x(\bm{x})=v[F_x](t,x)=\int I(x_{*}-x) F_x(t,x_{*})dx_{*}  
\end{equation}
and equivalent for $v_y$.
We can then rewrite the Master equation as
\begin{equation}
\begin{split}
\label{MasterEqWeak}
 \frac{d}{dt}\int \phi(\bm{x})f(t,\bm{x})d\bm{x}
 = 2h \int \Big(\frac{\p\phi}{\p x}(\bm{x})v_x(\bm{x})+\frac{\p\phi}{\p y}(\bm{x})v_y(\bm{x})
	\Big)f(t,\bm{x})d\bm{x}.
\end{split}
\end{equation}

This is the weak formulation of
\begin{equation}\label{1D_equation}
\frac{\partial f}{\partial t}+2h\text{div}(vf) = 0.
\end{equation}
For instance in the particular case $I(x)=x/|x|$ and $I(y)=y/|y|$, we have
\begin{equation} 
v_x(x) = v[F_x](t,x) = \int_{\{s>x\}} F_x(t,s)ds - \int_{\{s<x\}} F_x(t,s)ds, 
\end{equation}                                                                                               
and equivalent for the $y$ coordinate, which is exactly the vector field obtained in the 1D model in \cite{PSB}.

Since the interaction rule for uncorrelated model treats both coordinates $x$ and $y$ as independent, it is reasonable to conjecture that the marginals $F_x$ and $F_y$ will both evolve  following the 1D dynamic.

We prove this is indeed the case.
Take a test function $\phi=\phi(x)$ and notice that
$\int_{\mathbb{R}^2} \phi(x)f(t,\bm{x})d\bm{x}= \int_{\mathbb{R}} \phi(x)F_1(t,x)dx$.
Then \eqref{MasterEqWeak}  becomes
\begin{equation}\label{WeakFormMarg}
\begin{split}
h \frac{d}{dt} \int \phi(x)F_1(t,x)dx
&= \int d\bm{x}f(t,\bm{x}) \phi'(x)v_x(x) \\
&= \int F_1(t,x) \phi'(x)v_1(x)dx. 
\end{split}
\end{equation}

The same applies to $F_y$.
Thus $F_x$ and $F_y$ both solve the conservation law in $\mathbb{R}$ given by
\begin{equation} \label{EquationF}
 \frac{\partial F_x}{\partial t}+2h\frac{\partial}{\partial x}(v_x F_x) = 0 
\end{equation}
as expected. In view of \eqref{VF2}, we can rewrite \eqref{WeakFormMarg} as
\begin{equation}
\begin{split}\frac{d}{dt} \int \phi(x)F_k(t,x)dx
= \hspace{4cm}\\
\hspace{2cm} 2h \iint F_x(t,x) \phi'(x) G(x_{*}-x) F_x(t,x_*)dx_{*}dx.  
\end{split}
\end{equation}
Thus $F_x(t,.)$ satisfy Eq. \eqref{WeakFormApprox} obtained in the previous section
when studying the 2D dynamic. An equivalent reasoning is valid for $F_y(t,.)$. 
It follows that as $t\to +\infty$, $F_x(t)$ and $F_y(t)$ converges to the Dirac mass located at the mean value $m_x(0)$ and $m_y(0)$  of $F_x(0)$ and $F_y(0)$ respectively, i.e. $F_x(t)\to \delta_{m_x(0)}$ and where
\begin{equation}
m_x(0):=\int x F_x(0,x) dx = \int x f(0,\bm{x}) d\bm{x}. \end{equation}

and equivalent for $F_y(t)$ As a consequence
\begin{equation} 
 f(t,\bm{x}) \to \delta_{\bm{m}(0)} \qquad \text{as $t\to +\infty$,}\end{equation}

where $\bm{m}(0)=(m_x(0),m_y(0))$ is the initial mean opinion.

\end{appendices}

\end{document}